\documentclass[12pt]{article}

\usepackage{amsmath,amsthm,amssymb,amscd}
\usepackage{bbm}
\usepackage{graphicx, float}
\usepackage[hidelinks]{hyperref}
\usepackage[semicolon,comma]{natbib}
\usepackage{ds}

\usepackage{tabularx} 
\newcolumntype{L}[1]{>{\raggedright\arraybackslash}p{#1}} 

\usepackage{enumitem} 
\usepackage{caption, subcaption} 
\usepackage{capt-of} 
\usepackage[flushleft]{threeparttable} 
\usepackage{multirow}
\usepackage{booktabs}
\usepackage{xcolor,soul} 

\usepackage{authblk}

\newtheorem{Proposition}{Proposition} 

\title{Quantile estimation of CO$_2$ marginal abatement cost across emission-generating technologies}

\author[a]{Haleh Delnava}
\author[b,\footnote{Corresponding author. \newline
\hspace*{5mm} \textit{E-mail addresses:} \texttt{halehh.del@gmail.com (H. Delnava)}, \texttt{sheng.dai@zuel.edu.cn (S. Dai)}.}]{Sheng Dai}

\affil[a~]{Institute of Manufacturing Information \& System, National Cheng Kung University, Taiwan}
\affil[b~]{School of Economics, Zhongnan University of Economics and Law, 430073 Wuhan, China}
\date{\today}

\begin{document}
\captionsetup[figure]{labelfont={bf},labelformat={default},labelsep=period,name={Fig.}}
\captionsetup[table]{labelfont={bf},labelformat={default},labelsep=period,name={Table}}

\maketitle
 
\vfill

\vfill

\begin{abstract}
\noindent 
Marginal abatement cost (MAC) is a critical metric for designing efficient and cost-effective mitigation policies. However, existing MAC estimates are typically derived under different assumptions about emission-generating technologies, yet few studies have systematically compared these technologies. Moreover, conventional estimators often exhibit biases arising from limited abatement options, production inefficiencies, and data noise. To address these limitations, this paper analyzes the abatement behavior of three emission-generating technologies: by-production, joint disposability, and weak G-disposability, each consistent with the material balance principle. We employ both full and quantile frontier estimation methods to identify optimal abatement strategies. Using data from U.S. coal-fired power plants in 2022, the empirical results suggest that reducing electricity output, rather than cutting emission-generating inputs such as fossil fuels, provides a more cost-effective mitigation pathway. Furthermore, Monte Carlo simulations demonstrate that the quantile estimator consistently delivers more accurate results than the full frontier estimator.
\\[5mm]
\end{abstract}
\vfill

\thispagestyle{empty}

\newpage
\setcounter{page}{1}
\setcounter{footnote}{0}
\pagenumbering{arabic}
\baselineskip 20pt

%

\section{Introduction}\label{sec:intro}

Mitigation policies are generally required to balance environmental objectives with economic efficiency to achieve sustainable development \citep{Wu2023, Zhang2025}. The marginal abatement cost (MAC) plays a pivotal role in this balance by measuring the opportunity cost of reducing undesirable outputs relative to forgone desirable outputs. Accordingly, MAC not only helps identify least-cost abatement pathways but also informs the efficient allocation of mitigation responsibilities. To estimate MAC, two principal approaches are widely adopted: economic and engineering methods, which are regarded as complementary \citep{Lee2014}. The economic approach integrates diverse abatement technologies by modeling joint production based on observed input–output data \citep{Aiken2003}. In contrast, the engineering approach analyzes individual abatement technologies separately via MAC curves.

In the field of production economics, a range of approaches have been developed to model emission-generating technologies for the estimation of shadow prices and MAC \citep{Dakpo2019}. Following \citet{Zhou2008a}, neoclassical models typically adopt one of two strategies: treating undesirable outputs as inputs \citep[see, e.g.,][]{Hailu2001, Considine2006, Mahlberg2011} or applying data transformations that convert undesirable outputs into desirable ones via a reverse function \citep{Scheel2001}. Alternatively, undesirable outputs can be modeled directly under the assumptions of weak disposability (WD; \citealp{Fare1985, Fare2005a}) or null-jointness. WD implies a trade-off between desirable and undesirable outputs, where a reduction in undesirable output necessitates a reduction in desirable output. Null-jointness assumes that undesirable outputs are jointly produced with desirable outputs and thus cannot be zero unless the desirable outputs are also zero \citep[see, e.g.,][]{Fare1996, Coggins1996a, Boyd1999}. However, comparative analyses of these modeling approaches remain relatively scarce in the existing literature.

In practice, parametric programming \citep[see, e.g.,][]{Fare2005a, Wei2013a}, data envelopment analysis (DEA) \citep{Fare2014}, and stochastic nonparametric envelopment of data (StoNED) \citep{Mekaroonreung2012a, Lee2018b} are commonly utilized to estimate emission-generating technologies. Parametric programming approaches impose a predefined functional form (e.g., translog or quadratic) for the production, cost, or distance function, allowing for the derivation of shadow prices through differentiability. However, the reliability of estimates relies on the correctness of the assumed functional form. DEA, by contrast, is a nonparametric estimation and avoids functional form assumptions, but it neglects stochastic noise \citep{Vardanyan2006}. The StoNED method offers a unified framework that relaxes functional form assumptions while accounting for statistical noise \citep{Kuosmanen2012c}. 

While widely used in empirical studies, these full frontier approaches may be affected by several underlying factors that tend to cause systematic overestimation of MAC \citep{Kuosmanen2020b}. First, shadow prices are typically computed on the efficient full frontier in conventional analyses, overlooking the heterogeneous efficiency levels of decision-making units (DMUs). Second, the cost-effectiveness of major abatement strategies, including production downscaling, the adoption of low-carbon fuels (e.g., clean hydrogen, sustainable biofuels), and investments in pollution control technologies (e.g., carbon capture, utilization, and storage), in reducing carbon footprints is often neglected. Third, conventional estimation methods are sensitive to statistical noise, outliers, and extreme values.

To avoid overestimation of shadow prices and MAC, \citet{Kuosmanen2020b} develop a data-driven convex quantile regression (CQR) approach to estimate quantile shadow prices. Unlike conventional full frontier estimation, CQR estimates shadow prices within the production possibility set using nearest quantiles, thereby explicitly accounting for inefficiency. Moreover, the CQR approach inherits robustness to noise and outliers from quantile regression \citep{Koenker2005b}. This framework has been applied in various empirical contexts. For example, \citet{Kuosmanen2020} evaluate the total abatement cost of the Kyoto Protocol for OECD countries. \citet{Dai2020} estimate the MAC of CO$_2$ for Chinese provinces and observe that, under existing targets, some provinces might experience economic gains from moderate emission increases, highlighting opportunities for improving climate policy design. \citet{Zhao2022} apply CQR to U.S. coal power plants, documenting rising shadow prices for CO$_2$, SO$_2$, and NOx between 2010 and 2017, largely driven by electricity market dynamics and emission reduction policies.

Furthermore, modeling emission-generating technologies needs to satisfy the material balance principle (MBP), as neglecting this principle may result in misguided policy decisions \citep{Ayres1969}. To assess the consistency of such technologies with the MBP, \citet{Coelli2007}, \citet{Forsund2009}, and \citet{Hoang2011} mathematically demonstrate that the production models proposed by \citet{Fare1989} violate the first law of thermodynamics, which governs the conservation of mass and energy. Therefore, the MBP should be explicitly incorporated into the estimation of MAC for undesirable outputs.

To address the challenges inherent in shadow prices and MAC estimation, the first contribution of this paper is to compare three representative emission-generating technologies: by-production \citep{Murty2012}, joint disposability \citep{Ray2018}, and weak G-disposability \citep{Rdseth2017, Hampf2015}. These technologies are consistently formulated within the MBP framework. Crucially, we adopt a unified by-production technology framework proposed by \citet{Shen2021}, which directly addresses the widely noted criticism that the conventional by-production model neglects the critical linkage between sub-technologies \citep{Lozano2015, Dakpo2017}. In addition, the analysis incorporates both full and quantile frontier estimation techniques, offering a more robust and accurate approach for identifying optimal abatement strategies.

The second methodological contribution is to propose three sign-constrained convex nonparametric least square (CNLS) models to characterize the by-production, joint disposability, and weak G-disposability technologies, respectively. The equivalences between the conventional DEA-based model and the sign-constrained CNLS model have been stated and proved. More importantly, we develop the quantile models to estimate three emission-generating technologies. The shadow prices are then locally estimated and are more robust to outliers and extreme values.  

Our third contribution lies in an empirical analysis of U.S. coal-fired power plants in 2022, showing that reducing electricity output is more cost-effective for emission reduction than lowering fossil fuel inputs. Monte Carlo simulations confirm the superiority of the quantile models through significantly lower root mean squared error values. We also validate the practical applicability of different emission-generating technologies. Moreover, by examining plants with the highest and lowest MACs, we reveal the technological, operational, and fuel-related factors underlying cost variation, highlighting the policy relevance of estimator choice.

The rest of this paper is organized as follows. Section~\ref{sec:pollut} reviews the theoretical relationships among the three primary emission-generating technologies. Section~\ref{sec:shadow} explains the estimation of shadow prices using both full and quantile estimators for these technologies and discusses the identification of the least-cost MAC strategy. An application to U.S. coal-fired power plants and the Monte Carlo simulation are demonstrated in Sections~\ref{sec:emp} and~\ref{sec:monte}. Section~\ref{sec:concl} concludes the paper with policy implications.

%

\section{Emission-generating technologies}\label{sec:pollut}

In the production theory, the free disposability of inputs and outputs is a neoclassical assumption. However, this assumption can be violated in the presence of the undesirable output. Several alternative disposability approaches have thus been proposed to model emission-generating technologies that yield both desirable and undesirable outputs simultaneously. In fact, an ideal emission-generating technology should satisfy all the following properties: 1) positive correlation between emissions and emission-generating inputs; 2) positive correlation between inputs and desirable outputs; 3) positive correlations between desirable and undesirable outputs. 

Table~\ref{tab:tab1} presents the different emission-generating technologies, which are consistent with the MBP. Regardless of the emission-generating technology employed (i.e., the second row of Table~\ref{tab:tab1}), MAC is calculated as the ratio of multipliers of desirable and undesirable outputs in the relevant dual linear programming model, with revenue maximization serving as the objective function. However, this ratio signifies a reduction in the production of desirable outputs necessary per unit decrease in the generation of undesirable outputs, assuming that the production unit is technically efficient. However, in the presence of technical inefficiency, this ratio cannot be reliably interpreted as the MAC because it does not accurately represent the trade-off between desirable and undesirable outputs. To address this issue, quantile frontier approaches incorporating specified emission-generating technologies are developed to capture the impact of inefficiency (i.e., the third row of Table~\ref{tab:tab1}; see more discussion in Section~\ref{sec:shadow}).
\begin{table}
    \centering
    \caption{MAC estimation using different emission-generating technologies.}
    \begin{tabular}{cccc}
    \toprule
                                 & By-production  & Weak G-disposability & Joint disposability\\
    \midrule
        Full frontier          & \citet{Murty2012}  & \citet{Hampf2015} & \citet{Ray2018}  \\
                      & \citet{Shen2021}  & \citet{Rodseth2023}  &                  \\
        Quantile frontier      & This paper & This paper & This paper \\
    \bottomrule
    \end{tabular}
    \label{tab:tab1}
\end{table}

Suppose there are $I$ DMUs indexed by $i$, each DMU consists of $M$ inputs, $J$ desirable outputs, and $K$ undesirable outputs. The input and output vectors are denoted by $x \in \real_+^M$, $y \in \real_+^J$, and $b \in \real_+^K$, respectively. The input vector is divided into two sub-components, $x=(x^N, x^P)$, where $x^N$ represents $M_1$ non-emission-generating inputs, and $x^P$ is $M_2$ emission-generating inputs. For the sake of simplicity, the same notation is used throughout this paper.

\subsection{By-production technology}

Let $\Ts_{BP}$ be the by-production technology \citep{Murty2012}, which is an intersection of two sub-technologies: the economic technology ($\Ts_1$) and the environmental technology ($\Ts_2$).
\begin{alignat}{2}
\Ts_{BP}= & \{(x^N, x^P, y, b) \in \real_+^{M+S+J}| \lambda X \le x, \lambda Y \geq y, \mathbbm{1}^T\lambda=1, \mu X^P \geq x^p, \mu B \le b, \notag \\
& \mathbbm{1}^T \mu=1; \text{ for } \lambda \in \real_+^i, \mu \in \real_+^i\} = \Ts_1 \cap \Ts_2, \notag 
\end{alignat}
\vspace{-0.5cm}where
\begin{alignat}{2}
   & \label{eq:eq1}  \Ts_1=\{(x,y) \in \real_+^{M+S}| \lambda X \le x, \lambda Y \geq y, \mathbbm{1}^T \lambda=1; \text{ for } \lambda \in \real_+^i\}  \\
   & \label{eq:eq2}  \Ts_2=\{ (x^P,b) \in \real_+^{M_2+J}| \mu X^P \geq x^p, \mu B \le b, \mathbbm{1}^T \mu=1; \text{ for } \mu \in \real_+^i\}
\end{alignat}

The sub-technology $\Ts_1$~\eqref{eq:eq1} defines economic production technology under the assumption of variable returns to scale (VRS). In this context, $\lambda$ denotes the intensity variables used in the convex combination of all inputs and desirable outputs. This sub-technology aligns with standard neoclassical disposability properties, specifically the free-disposability of both desirable outputs and all inputs.
$$(x^N, x^P, y, b) \in \Ts_1 \land ( \bar{x}^N \geq x^N \land \bar{x}^P \geq x^P) \land \bar{y} \leq y \Rightarrow (\bar{x}^N, \bar{x}^P, \bar{y}, b) \in \Ts_1.$$

The sub-technology $\Ts_2$~\eqref{eq:eq2} defines environmental (or residual-generating) technology under the assumption of VRS. Here, $\mu$ represents the intensity variables used in the convex combination of emission-generating inputs and undesirable outputs. It assumes reliance on the costly disposability of emissions and emission-generating inputs. Costly disposability of emissions suggests a minimum level of emissions corresponding to specific amounts of emission-generating inputs. Similarly, the costly disposability of these inputs (e.g., fuel) implies that for any given fixed level of emissions, there is a maximum amount of emission-generating inputs.
$$(x^N, x^P, y, b) \in \Ts_2 \land  \bar{x}^P \leq x^P \land \bar{b}^b \geq b \Rightarrow (x^N, \bar{x}^P, y, \bar{b}) \in \Ts_2.$$

\subsection{Weak G-disposability}

The weak G-disposability highlights the value of viewing technical interactions between resources, residuals, and good output, which is more flexible than conventional joint production of good and undesirable outputs (i.e., WD). This framework applies directional distance function (DDF) \citep{Chambers1996, Chambers1998b} to integrate mass/energy conservation with G-disposability \citep{Hampf2015,Rodseth2023}. Under the condition of weak G-disposability, the production possibility set is characterized by the following assumptions.
\begin{enumerate}[label=A\arabic*.]
    \item $\Ts_{WGD}$ is convex.
    \item \label{lab2} Output essentiality for the undesirable outputs: If $(x,y,b) \in \Ts_{WGD} \land b=0 \Rightarrow x^P=0$. 
    \item \label{lab3} Input essentiality for the undesirable outputs: If $(x,y,b) \in \Ts_{WGD} \land x^P=0 \Rightarrow b=0$. 
	\item \label{lab4} Inputs and outputs are weakly G-disposable: If $(x,y,b) \in \Ts_{WGD} \land ug_x+rg_y-g_b=0 \Rightarrow (x^p+g_x,y-g_y,b+g_b) \in \Ts_{WGD}$, where $g_{(.)}$ denote pre-assigned direction vector. 
\end{enumerate}

Assumptions~\ref{lab2} and~\ref{lab3} are consistent with the second law of thermodynamics, which states that the use of polluting inputs to meet energy requirements in any production process inevitably generates residuals. That is, the production process is zero-emission if no emission-generating input is utilized. Assumptions~\ref{lab2} and~\ref{lab3} are extensions of the null-jointless assumption \citep{Fare2012}, where the presence of polluting inputs is overshadowed. 

Assumption~\ref{lab4} ensures compliance with the laws of conservation, where $u$ and $r$ represent the emission factors of emission-generating inputs and the recuperation factor of desirable outputs, respectively. Note that emission factors for non-emission-generating inputs (e.g., labor and capital) equal zero, as they do not contribute to emissions \citep{Coelli2007, Lauwers2009}. This assumption mandates the disposability of inputs and outputs subject to a summation constraint, indicating that any increase in pollution due to higher input use and/or reduced desirable output must be offset by a corresponding rise in undesirable outputs during the disposal process.

The WGD technology is formulated as
\begin{alignat}{2}
\label{eq:eq3}
    \Ts_{WGD}=& \{(x^N, x^P, y, b) \in \real_+^{M+S+J}| \lambda X^N \le x^N, \lambda X^P + s_x = x^P, \lambda Y- s_y = y, \lambda B + s_b = b, \\ 
    & u s_x + r s_y - s_b=0, \mathbbm{1}^T \lambda=1; \text{ for } \lambda \in \real_+^i \}, \notag
\end{alignat}
where the selection of direction vectors $(g_x, g_y, g_b)$ are replaced by their empirical counterparts, the slacks $(s_x,s_y,s_b)$.

\subsection{Joint disposability}

Analogous to the BP and weak G-disposability paradigms, joint disposability (JD) explicitly differentiates between emission-generating inputs and non-emission-generating inputs \citep{Ray2018}. The following assumptions on the production possibility set in the JD framework are considered.
\begin{enumerate}[label=A\arabic*.]
    \setcounter{enumi}{4}
    \item Free disposability of  inputs and desirable outputs 
    \item[] If $ (x^N,x^P,y) \in \Ts_3, \bar{x}^N \geq x^N, \bar{x}^P \geq x^P, \bar{y} \le y \Rightarrow (\bar{x}^N, \bar{x}^P, \bar{y}) \in \Ts_3.$
    \item Weak disposability between emission and emission-generating inputs 
    \item[] If $(x^P,b) \in \Ts_4, 0 \le \theta \le 1, \Rightarrow (\theta x^P,\theta b) \in \Ts_4.$
    \item $\Ts_{JD}=\Ts_3 \cap \Ts_4.$
\end{enumerate}

The disposability assumption for original JD technology posits the WD assumption of emissions and emission-generating inputs in residual-generating technology ($\Ts_4$), while assuming free disposability of non-emission-generating inputs and desirable outputs in production technology ($\Ts_3$). This approach distinguishes between consolidated and decentralized methods, where a single intensity vector (i.e., $\lambda$) is used in the former method to construct reference input–output bundles in both sub-technologies, and two distinct intensity vectors (i.e., $\lambda$ for $\Ts_3$ and $\mu$ for $\Ts_4$) are utilized in the latter. Recognizing the interdependence of processes like electricity generation and pollution emission, the specification of two separate intensity variables is questionable. Therefore, benchmark bundles are created by treating the entire input–output bundle as a single peer group vector.
\begin{alignat}{2}
\label{eq:eq4}
\Ts_{JD} = \{(x^N, x^P, y, b) \in \real_+^{M+S+J}| \lambda X^N \le x^N, \lambda X^P=x^P, \lambda Y \geq y, \lambda B=b, \mathbbm{1}^T\lambda=1; \text{ for } \lambda \in \real_+^i\}.
\end{alignat}

%

\section{Shadow price estimation}\label{sec:shadow}

\subsection{Full frontier estimation}\label{sec:deter}
\subsubsection{BP technology}

As discussed in Section~\ref{sec:pollut}, the emission-generating technology should simultaneously satisfy three fundamental assumptions. However, conventional BP technology fails to represent the trade-off between desirable and undesirable outputs. Overlooking such a crucial assumption can result in biased estimates, potentially leading to misguided environmental policymaking. In this paper, we not only model the positive trade-off between desirable and undesirable outputs but also demonstrate the reduction of undesirable outputs through the decreased use of emission-generating inputs. 

Under the BP technology with the VRS specification, we employ the DDF to characterize the simultaneous changes in the entire input-output space and construct the following graph efficiency improvement model.
\begin{alignat}{2}
	\underset{ } {\mathop{\max }}\, \quad & w_1\theta_m\ +w_2\theta_j+w_3\theta_k &{\quad}& \label{eq:eq5}\\
	\mbox{s.t.}\quad
    & \sum_{i=1}^{I} \lambda_i y_{ij} \geq y_{oj} +\theta_j g_y &{}& \forall j \notag \\
	& \sum_{i=1}^{I}\lambda_i x_{im}^N \le x_{om}^N  &{}& \forall m_1 \notag \\
	& \sum_{i=1}^{I}\lambda_ix_{im}^P \le x_{om}^P-\theta_m g_x &{}& \forall m_2 \notag  \\ 
	& \sum_{i=1}^{I} \mu_ib_{ik} \le b_{ok}-\theta_k g_b &{}& \forall k \notag  \\
    & \sum_{i=1}^{I}\lambda_ix_{im}^P=\sum_{i=1}^{I}\mu_ix_{im}^P &{}& \forall m_2 \notag  \\
    & \sum_{i=1}^{I}\lambda_i=1,  \,  \sum_{i=1}^{I}\mu_i=1  &{}&   \notag  \\
    & \lambda_i \geq 0, \, \mu_i \geq 0 &{}& \forall i  \notag  \\
    & \theta_m,  \,  \theta_j,  \,  \theta_k \geq 0  &{}& \forall m_2, j, k \notag 
\end{alignat}
where $(g_x,g_b, g_y)$ represents the non-zero directional vector associated with the adjustments in emission-generating inputs, undesirable outputs, and desirable outputs (while keeping non-emission-generating inputs constant). $\lambda_i$ and $\mu_i$ are the intensity variables for $\Ts_1$ and $\Ts_2$, respectively. 

Model~\eqref{eq:eq5} captures efficiency improvement that maximizes the weighted sum of proportional reductions in emission-generating input ($\theta_m$), undesirable output ($\theta_k$), and proportional increase in the desirable output ($\theta_j$). The weights assigned to proportionate changes in emission-generating inputs, desirable and undesirable outputs are $w_1 \in (0,1)$, $w_2\in (0,1)$, and $w_3\in (0,1)$, respectively. Under the assumption of free disposability of inputs, the optimal solution of model~\eqref{eq:eq5} indicates no efficiency improvement in the usage of emission-generating inputs when these inputs are assigned zero weights ($\theta_m=0$ whenever $w_1=0$; $T_{WD}$). In contrast, under the BP and joint disposability technologies, only non-emission-generating inputs are freely disposable, leading to inefficiencies, $\theta_m\neq0$, even when the weights for emission-generating inputs are set to zero. Accordingly, we assign a zero weight to the proportionate change in the emission-generating inputs ($w_1 = 0$), while equal weights are allocated to the proportionate changes in both desirable and undesirable outputs ($w_2 = w_3 = 0.5$).
\begin{Proposition}
    The graph efficiency improvement model~\eqref{eq:eq5} can be equivalently reformulated as the following sign-constrained CNLS model~\eqref{eq:eq6}.
    \vspace{-0.3cm}
\begin{alignat}{2}
	\underset{ } {\mathop{\min }}\, \quad & \sum_{i=1}^{I} \varepsilon_i^2 &{\quad}& \label{eq:eq6}\\
	\mbox{s.t.}\quad
    & \varepsilon_i=(\alpha_i-{\bar{\alpha}}_i)+\beta_i^\prime x_{im}^N+\eta_i^\prime x_{im}^p+\omega_i^\prime b_{ik}-\gamma_i^\prime y_{ij} &{}& \forall i\notag \\
	& \alpha_h{+\beta}_h^\prime x_{im}^N+(\eta_h^\prime+{\bar{\eta}}_h^\prime{)x}_{im}^p-\gamma_h^\prime y_{ij}\le\alpha_i{+\beta}_i^\prime x_{im}^N+(\eta_i^\prime+{\bar{\eta}}_i^\prime{)x}_{im}^p-\gamma_i^\prime y_{ij}  &{}& \forall i, h \notag \\
    & \omega_h^\prime b_{ik}-(\bar{\alpha}_h+\bar{\eta}_h^\prime x_{im}^p)\le\ \omega_i^\prime b_{ik}-(\bar{\alpha}_i+\bar{\eta}_i^\prime x_{im}^p) &{}& \forall i, h \notag \\
	& \gamma_i^\prime g_y \geq 0.5, \, \omega_i^\prime g_b \geq 0.5,  \, \eta_i^\prime g_x \geq 0 &{}& \notag \\ 
    & \gamma_i\geq 0,  \, \beta_i\geq 0,  \, \eta_i\geq 0,  \, \omega_i\geq 0 &{}& \forall i \notag \\
    & \varepsilon_i\geq0  &{}& \forall i \notag  
\end{alignat}
where $\beta_i$, $\eta_i$ and $\gamma_i$ represent the corresponding dual variables of $x_i^{NP}$, $x_i^P$, and $y_i$ under sub-technology $\Ts_1$. $\bar{\eta}_i$ and $\omega_i$ denote the multipliers of $x_i^P$ and $b_i$ under sub-technology $\Ts_2$. The multiplier $\alpha_i$ and $\bar{\alpha}_i$ define the intercepts in constructing DMU-specific hyperplanes for two sub-technologies $\Ts_1$ and $\Ts_2$, respectively. Recall that, as proportional changes in emission-generating inputs and outputs are assumed to be positive in the primal model~\eqref{eq:eq5}, the corresponding constraints are formulated as inequalities. The formulation~\eqref{eq:eq6} implies that the shadow prices associated with these constraints (i.e., $\gamma$, $\eta$, $\omega$) lead to non-identical values. 
\end{Proposition}
\begin{proof}
See Appendix~\ref{App:proof1}.
\end{proof}

\subsubsection{Weak G-disposability}

Under WGD, the selection of the direction vector is constrained by the summing-up constraint. However, to estimate the positive trade-off between emission-generating inputs and undesirable outputs and between desirable and undesirable outputs, \citet{Rodseth2023} propose a slack-based DDF method with fixed direction vectors $(g_x, g_y, g_b) = (1, 1, 1)$ and relaxes the impact of the recuperation factor of desirable outputs by setting $r = 0$ in the summing-up constraint. The weak G-disposability under sign-constrained CNLS model is formulated as
\begin{alignat}{2}
	\underset{ } {\mathop{\min }}\, \quad & \sum_{i=1}^{I} \varepsilon_i^2 &{\quad}& \label{eq:eq7}\\
	\mbox{s.t.}\quad
    & \gamma_i^\prime y_i=\alpha_i+\beta_i^\prime x_i^N+\eta_i^\prime x_i^P+\omega_i^\prime b_i-\varepsilon_i &{}& \forall i\notag \\
	& \alpha_i+\beta_i^\prime x_i^N+\eta_i^\prime x_i^P+\omega_i^\prime b_i-\gamma_i^\prime y_i\le\alpha_h+\beta_h^\prime x_i^N+\eta_h^\prime x_i^P+\omega_h^\prime b_i-\gamma_h^\prime y_i  &{}& \forall i, h \notag \\
	& \eta_i^\prime+\omega_i^\prime u\geq0  &{}& \forall i  \notag \\ 
    & \gamma_i^\prime+\omega_i^\prime+\eta_i^\prime=1 &{}& \forall i \notag \\
    &\beta_i\geq0,\gamma_i\geq0  &{}& \forall i \notag  \\
    & \varepsilon_i\geq0  &{}& \forall i \notag  
\end{alignat}
where $\beta_i$, $\eta_i$, $\omega_i$, and $\gamma_i$ characterize the corresponding dual variables of $x_i^{NP}$, $x_i^P$, $b_i$ and $y_i$, respectively. The objective function in model~\eqref{eq:eq7} calculates the sum of squared disturbance terms. The first set of constraints denotes the distance to the frontier as a linear function of inputs and outputs. The second set of constraints ensures convexity among the hyperplanes in all pairs of observations \citep{Afriat1972}. The third set of constraints imposes the WGD constraint that addresses the correlation between the dual prices of polluting inputs and pollution, contingent upon the material flow coefficient, $u$ (i.e., $us_x-s_b=0$). The fourth set of constraints describes the translation property \citep{Chambers1998b}, and non-negativity of dual variables retains the monotonic production frontier. Furthermore, in contrast to the specification of weak disposability, inputs are not freely disposable in model~\eqref{eq:eq7}.

However, model~\eqref{eq:eq7} requires the DMU-specific knowledge on emission and recuperation factors (i.e., $u$, $r$). For instance, the recuperation factor ($r$) for electricity generation is calculated as the ratio of recoverable electricity generation to gross electricity generation, multiplied by 100. Moreover, different types of coal (e.g., bituminous, sub-bituminous, lignite) have distinct emission factors due to variations in their composition and energy content. In addition, the WGD framework imposes highly restrictive constraints on the reduction of undesirable outputs due to the summing-up condition. This constraint necessitates trade-offs, meaning that reducing an undesirable output often requires either maintaining a fixed level of inputs or a fixed level of desirable outputs, thereby limiting flexibility and potential efficiency improvement.

\subsubsection{Joint-disposability}

In the context of WD, any commodity (output or input) should be analyzed in conjunction with other elements. The approach introduced by~\citet{Fare1989} highlights that, under WD, both desirable and undesirable outputs can be proportionally reduced while maintaining constant input levels. Moreover, this approach entails a trade-off between emission-generating inputs and emissions, while assuming standard free disposability for desirable outputs and non-emission-generating inputs.

We utilize DDF to characterize JD technology under the VRS specification and have the following JD model
\vspace{-0.79cm}
\begin{alignat}{2}
	\underset{ } {\mathop{\max }}\, \quad & \theta &{\quad}& \label{eq:eq8}\\
	\mbox{s.t.}\quad
    & \sum_{i=1}^{I}{\lambda_iy_{ij}\geq y_{oj}}+\theta g_y &{}& \forall j \notag \\
	& \sum_{i=1}^{I}{\lambda_ix_{im}^N\le x_{om}^N}  &{}& \forall m  \notag \\
	& \sum_{i=1}^{I}{\lambda_ix_{im}^P=x_{om}^P}-\theta g_x &{}& \forall j  \notag  \\ 
    & \sum_{i=1}^{I}{\lambda_ib_{ik}=}b_{ok}-\theta g_b &{}& \forall k  \notag  \\ 
    & \sum_{i=1}^{I}\lambda_i=1  &{}& \forall i  \notag  \\ 
    & \lambda_i\geq0, \theta\geq0 &{}& \forall i  \notag 
\end{alignat}
where the DMU is identified as efficient if the evaluated value $\theta$ is zero, and DMUs that meet feasibility criteria but are inefficient will exhibit values exceeding zero.
\begin{Proposition}
    The JD model~\eqref{eq:eq8} can also be equivalently reformulated as the sign-constrained CNLS model~\eqref{eq:eq9}.
\end{Proposition}
\vspace{-1cm}
\begin{alignat}{2}
	\underset{ } {\mathop{\min }}\, \quad & \sum_{i=1}^{I} \varepsilon_i^2 &{\quad}& \label{eq:eq9}\\
	\mbox{s.t.}\quad
    & \gamma_i^\prime y_i{=\alpha}_i+\beta_i^\prime x_i^N+\eta_i^\prime x_i^P+\omega_i^\prime b_i-\varepsilon_i  &{}& \forall i\notag \\
	& \alpha_i+\beta_i^\prime x_i^N+\eta_i^\prime x_i^P+\omega_i^\prime b_i-\gamma_i^\prime y_i\le\alpha_h +\beta_h^\prime x_i^N+\eta_h^\prime x_i^P+\omega_h^\prime b_i-\gamma_h^\prime y_i  &{}& \forall i, h \notag \\
    & \eta_i^\prime g_x+\omega_i^\prime g_b+\gamma_i^\prime g_y=1 &{}& \forall i \notag \\
    &\beta_i\geq0,\gamma_i\geq0  &{}& \forall i \notag  \\
    & \varepsilon_i\geq0  &{}& \forall i \notag  
\end{alignat}
\vspace{-1cm}
\begin{proof}
    See Appendix 1 in \citet{Kuosmanen2006}. 
\end{proof}

This section focuses on estimating shadow prices using nonparametric full frontier approaches across three emission-generating technologies. However, the critical challenge existing in these approaches is the considerable diversity of socioeconomic characteristics among real applications, which introduces unobserved heterogeneity and various forms of statistical noise. Such heterogeneity is often misconstrued as inefficiency by full frontier estimators. The impact of heterogeneity is particularly pronounced in panel data settings. Therefore, shadow price estimation requires a more robust approach to accurately account for unmeasured heterogeneity and inefficiency.

\subsection{Quantile frontier estimation}\label{sec:robu}

For a given quantile $\tau \in (0,1)$, the conditional nonparametric quantile function $Q_{y_i}(\tau| x_i, b_i)$ is given as \citep{Wang2014c, Kuosmanen2020b}
\begin{equation} \label{eq:eq10}
Q_{y_i}(\tau| x_i, b_i) = f(x_i, b_i) \times F_{\varepsilon_i}^{-1}(\tau) \quad 
\end{equation}
where $f(x_i, b_i)$ denotes the emission-generating function and $F_{\varepsilon_i}^{-1}$ refers to the inverse cumulative distribution function of the error term $\varepsilon_i$. In the multiple-input multiple-output settings, we can use the DDF to characterize a specific quantile $\tau$ of the production frontier
\begin{equation} \label{eq:eq11}
\overrightarrow{D}(x_i, y_i, b_i; g_x, g_y, g_b) = \sup\{\varepsilon | \prob(x_i^*, y_i^*,b_i^*) \geq 1 - \tau\}
\end{equation}
where the direction vector $(g_x, g_y, g_b) \in \mathbb{R}_{+}^{M_2+J+K}$ plays a crucial role in projecting inefficient DMUs onto the efficient frontier through the scaling of a composite error term. The optimal solution is then characterized by $(x_i^*, y_i^*, b_i^*)$. Note that the DDF inherently possesses fundamental axiomatic properties.

To obtain the unique emission-generating function, we employ an indirect quantile estimation approach,\footnote{
    Convex quantile regression (CQR) is a direct approach to estimating the quantile emission-generating function, as the quantile and expectile can be transformed into each other. See further detailed comparisons on CQR and CER approaches in \citet{Dai2023}.
}
convex expectile regression (CER), to estimate the emission-generating function $f(x_i, b_i)$. This is achieved by replacing the objective function and error term in the sign-constrained CNLS models \eqref{eq:eq6}, \eqref{eq:eq7}, or \eqref{eq:eq9}. For instance, under the BP technology, the sign-constrained CNLS model~\eqref{eq:eq6} is reformulated within the CER framework as follows
\begin{alignat}{2}
    \underset{ } {\mathop{\min }}\, \quad & (1-\tau) \sum_{i=1}^{I}(\varepsilon_i^+)^2 + \tau \sum_{i=1}^{I}(\varepsilon_i^-)^2 &{\quad}& \label{eq:eq12}\\
    \mbox{s.t.}\quad
    & \varepsilon_i^+ - \varepsilon_i^- =(\alpha_i-{\bar{\alpha}}_i)+\beta_i^\prime x_{im}^N+\eta_i^\prime x_{im}^p+\omega_i^\prime b_{ik}-\gamma_i^\prime y_{ij} &{}& \forall i\notag \\
	& \alpha_h{+\beta}_h^\prime x_{im}^N+(\eta_h^\prime+{\bar{\eta}}_h^\prime{)x}_{im}^p-\gamma_h^\prime y_{ij}\le\alpha_i{+\beta}_i^\prime x_{im}^N+(\eta_i^\prime+{\bar{\eta}}_i^\prime{)x}_{im}^p-\gamma_i^\prime y_{ij}  &{}& \forall i, h \notag \\
    & \omega_h^\prime b_{ik}-(\bar{\alpha}_h+\bar{\eta}_h^\prime x_{im}^p)\le\ \omega_i^\prime b_{ik}-(\bar{\alpha}_i+\bar{\eta}_i^\prime x_{im}^p) &{}& \forall i, h \notag \\
	& \gamma_i^\prime g_y \geq 0.5, \, \omega_i^\prime g_b \geq 0.5,  \, \eta_i^\prime g_x \geq 0 &{}& \notag \\ 
    & \gamma_i\geq 0,  \, \beta_i\geq 0,  \, \eta_i\geq 0,  \, \omega_i\geq 0 &{}& \forall i \notag \\
    & \varepsilon_i^+\geq0,  \, \varepsilon_i^-\geq0  &{}& \forall i \notag  
\end{alignat}
where the error terms ($\varepsilon_i^-$ and $\varepsilon_i^+$) denote negative and positive deviations from the quantile frontier $\tau$ and $1-\tau$ characterize the weights of error terms and segment the production possibility set into upper and lower quantiles. In the context of emission-generating technologies, the sign-constrained CNLS models determine the conditional mean by minimizing a squared error term through quadratic programming. In contrast, the CER model addresses the conditional quantile by minimizing asymmetric squared error terms, resulting in a unique emission-generating function.

The CER model~\eqref{eq:eq12} provides more general estimation as sign-constrained CNLS models are equivalent to upper $\tau$-efficient when $\tau$ approaches unity. In addition, at the optimum, $\varepsilon_i^-\times\varepsilon_i^+=0$, which means an observation can either have a positive quantile residual ($\varepsilon_i^+ \ge 0$) or a negative quantile residual ($\varepsilon_i^- \ge 0$), but not both simultaneously \citep{Kuosmanen2020b}. Recall that the CER model builds upon CQR, specifically by replacing the $L_1$norm distance found in CQR with a quadratic $L_2$ norm term in the objective function to guarantee uniqueness of the optimal solution.

From a statistical perspective, CER offers several advantages over sign-constrained CNLS models. The CER approach is more robust to outliers in the data space due to that it focuses on quantiles rather than the mean, making it particularly useful for datasets with skewed distributions or heavy tails. CER also provides a more comprehensive analysis of the relationship between input and output variables by examining various points in the distribution. Furthermore, unlike conventional frontier methods that are sensitive to the direction vector, CER is less dependent on the choice of the direction vector due to the smaller fraction of the distance to the $\tau$-frontier of inefficient units \citep{Kuosmanen2020b, Dai2023}.

\subsection{MAC estimation}\label{sec:mac}

To estimate shadow prices of undesirable outputs, the duality relationship between the DDF and the revenue function is applied \citep{Fare2005a}. MAC is then typically measured as the marginal rate of transformation (MRT) between the undesirable output and the desirable output,
 $$\text{MRT} = \dfrac{{\partial\vec{D}(x,y,b,\vec{g})}/{\partial b}}{{\partial\vec{D}(x,y,b,\vec{g})}/{\partial}y}.$$ 

However, in practice, DMUs have various alternative options, such as input substitution, investment in abatement technologies, demand reduction, and purchasing emission allowances, to reduce undesirable outputs \citep{Forsund2009}. These options are beyond the conventional trade-off between desirable and undesirable outputs in production. Consequently, the use of the MRT to represent MAC has been criticized in the literature, as it requires additional assumptions regarding abatement options.

In MAC analysis, the term ``bang for the buck'' commonly refers to maximizing the reduction of the undesirable output (e.g., carbon emission) at minimal cost, and serves as a benchmark for evaluating and prioritizing the cost-effectiveness of various reduction methods. Accordingly, \citet{Kuosmanen2020b} propose a novel approach to estimate MAC based on the least-cost principle
$$\text{MAC} = \min\{p \times \text{MRT},w \times \text{MP}\},$$
where $p$ and $w$ denote the market prices of desirable output and emission-generating input. MP refers to the marginal product of undesirable output with respect to the emission-generating input, and is calculated as 
$$MP=\dfrac{{\partial\vec{D}(x,y,b,\vec{g})}/{\partial b}}{{\partial\vec{D}(x,y,b,\vec{g})}/{\partial x}}.$$

The MP metric can be used to control undesirable outputs by measuring the impact of incrementally decreasing an emission-generating input while keeping other non-emission-generating inputs constant. After considering input-side abatement options (e.g., fuel switch or investment in cleaner technology), the MAC estimates substantially decrease; see the empirical application to U.S. electric power plants in \citet{Kuosmanen2020b}. The new approach to estimating MAC can help policymakers and businesses identify where they can achieve the most significant reduction in emissions at the least cost. This is crucial for designing effective environmental policies and strategies.
 
In the realm of shadow price estimation, DDF facilitates the simultaneous adjustment of inputs and outputs via a predefined direction vector. However, the MRT and MP estimates are sensitive to the choice of the direction vector, particularly in the full frontier estimation \citep{Lee2002}. \citet{Layer2020} propose a data-driven approach for selecting the direction vector, known as the radial mean squared error measure, which identifies a direction orthogonal to the true frontier at the center of the data cloud. It has demonstrated superior performance relative to alternative direction vectors, particularly in terms of estimation accuracy. This approach adapts flexibly to both the shape of the production technology (whether convex or concave) and the underlying data distribution. By accounting for noise in all inputs and outputs, it offers a more robust inefficiency assessment than conventional methods, which often overlook this aspect. Following \citet{Layer2020}, we normalize the emission-generating input, desirable, and undesirable outputs
\begin{alignat}{2}
& \check{x}_i^P=\frac{x_h^p-\min_h(x_h^p)}{\max_h(x_h^p)-\min_h(x_h^p)}  &{\quad}& \forall i, h  \\
& \check{b}_i=\frac{b_i-\min_h(b_h)}{\max_h(b_h)-\min_h(b_h)}  &{\quad}& \forall i, h \notag \\
& \check{y}_i=\frac{y_i-\min_h(y_h)}{\max_h(y_h)-\min_h(y_h)}  &{\quad}& \forall i, h \notag
\end{alignat}     
We then specify directional vectors as in Table~\ref{tab:tab2} for each emission-generating technology.
\begin{table}[H]
    \centering
    \caption{Specifying direction vector.}
    \begin{tabular}{lccc}
        \toprule
        Technology & $g_y$ & $g_b$ & $g_x$\\
         \midrule
        BP & $1-\text{median}({\check{y}}_i)$ & $\text{median}({\check{b}}_i)$ & $1-\text{median}({\check{x}}_i^P)$\\
        WGD & $s_y=0$	 & \multicolumn{2}{c}{$us_x-s_b=0$} \\
        JD & $1-\text{median}({\check{y}}_i)$ & $\text{median}({\check{b}}_i)$ & $1-\text{median}({\check{x}}_i^P)$\\
        \bottomrule
    \end{tabular}
    \label{tab:tab2}
\end{table}

We would stress that the CER or CQR models are very robust to the choice of the direction vector. For the quantile estimation, we consider seven quantiles $\tau = \{0.05, 0.20, 0.35, 0.50, 0.65, 0.80, 0.95\}$. When a DMU lies beyond the 95th quantile or below the 5th quantile, the nearest quantile is used to compute MRT and MP. In contrast, DMUs that fall within two quantiles (e.g., between the 35th and 50th quantiles) are evaluated using the mean value of the corresponding estimates at quantile $\tau$ and $\tau + 0.15$ \citep{Kuosmanen2020b, Dai2025}. The 5th percentile provides insight into the lower tail of the distribution, identifying the least efficient units. The median (50th percentile) represents the central tendency of the shadow price distribution, while the 95th percentile reflects the upper bound, capturing information on the most efficient units.

%

\section{Empirical application}\label{sec:emp}

\subsection{Data and variables}

We conduct an empirical application to estimate the MAC of CO$_2$ emissions for U.S. coal-fired power plants operating in 2022. We select medium- and large-sized coal-fired power plants, as benchmarking smaller units is challenging due to unstable ramp-up conditions. Moreover, the likely homogeneity in technology among smaller units could lead to underestimation in shadow prices. This paper specifically concentrates on bituminous coal, the predominant coal type used in the sector. Note that bituminous coal is primarily consumed in industrial applications that require high thermal energy, such as electricity generation and coke production in the steel industry.

Similar to the inputs and outputs selected in \citet{Mekaroonreung2012a}, \citet{Hampf2014}, and \citet{Walheer2020}, our dataset comprises one desirable output, one undesirable output, one emission-generating input, and two non-emission-generating inputs. The desirable output is the net electricity generation (100 MWh), and the undesirable output is CO$_2$ emissions (1000 tons). The total fuel consumption (1000 MMBtu) represents the emission-generating input, and the non-emission-generating inputs include plant nameplate capacity (MW) and plant operating availability (hours). We collect data on CO$_2$ and operating availability from the Clean Air Markets Program Data maintained by the U.S. Environmental Protection Agency (EIA). We aggregate the operating availability data from the unit level to the power plant level because each power plant comprises multiple units (e.g., boiler, turbine, generator, cooling system, control, and monitoring system). The fuel consumption and net generation are obtained from EIA-923. Table~\ref{tab:tab3} summarizes the descriptive statistics of 71 bituminous coal-fired power plants in 2022. 
\begin{table}[H]
  \centering
  \caption{Statistics for bituminous coal-fired power plants in 2022 ($I = 71$ DMUs)}
    \begin{tabular}{llrrrr}
    \toprule
    \multicolumn{1}{l}{Variable} & \multicolumn{1}{l}{Unit}  & \multicolumn{1}{c}{Mean}  & \multicolumn{1}{c}{Std. Dev.} & \multicolumn{1}{c}{Min}   & \multicolumn{1}{c}{Max} \\
    \midrule
    Electricity             & 100 MWh       & 43202.95  & 35398.83  & 1080.90   & 157010.82 \\
    CO$_2$ emissions        & 1000 tons     & 4473.20   & 3354.28   & 169.76    & 12337.61 \\
    Total fuel consumption  & 1000 MMBtu    & 44701.17  & 34007.28  & 1664.67   & 127240.77 \\
    Operating availability  & hours         & 14253.16  & 12548.21  & 937.12    & 75645.01 \\
    Nameplate capacity      & MWh           & 1345.84   & 838.1     & 229       & 3498.60 \\
    Electricity price       & \$/100 MWh    & 1.08      & 0.32      & 0.51      & 2.13 \\
    Fuel price              & \$/1000 MMBtu & 46677     & 2451.79   & 1798.00   & 14237.50 \\
    \bottomrule
    \end{tabular}%
  \label{tab:tab3}%
\end{table}%

Electricity prices (in 1000 \$/MWh) for each utility are obtained from the EIA-861 report. Following the methodology proposed by \citet{Mekaroonreung2012a}, plant-level electricity prices are derived by averaging the prices of electricity sold to both end-use customers and for resale by each utility. In cases where utility-specific price data are unavailable, the state-level average retail electricity price reported by EIA is used as a proxy. The average sales price of bituminous coal (in dollars per ton) is sourced from the Annual Coal Report for each relevant state. We assign coal prices at the plant level based on the state in which each plant is located. For states not included in the report, the national average coal price for the study year (i.e., \$4898/1000 MMBtu) is applied.

\subsection{Results and discussion}

We apply BP, JD, and WGD technologies to the U.S. power plants sample to illustrate the differences between the two modeling frameworks. The first is the full frontier estimation based on sign-constrained CNLS (model~\eqref{eq:eq6}, \eqref{eq:eq7}, \eqref{eq:eq9}), where inefficient units are projected onto the full frontier under the assumption of uniform inefficiency across all DMUs. The second is the quantile frontier estimation based on CER (model~\eqref{eq:eq12} and similar), in which shadow prices are estimated locally using nearest quantiles. We adapt the radial mean squared error method discussed in Section~\ref{sec:mac} to determine the direction vector, which is specified as $(g_x, g_b, g_y) = (0.71, 0.31, 0.78)$. These quadratic programming models are solved by the CPLEX solver in GAMS software. Table~\ref{tab:tab4} reports the mean and median estimates  $\binom{\text{mean}}{\text{median}}$ of MAC, $p \times \text{MRT}$, and $w \times \text{MP}$.
\begin{table}[H]
  \centering
  \caption{Mean and median of MAC and its components under full and quantile frontier estimators.}
    \begin{tabular}{lccccccc}
    \toprule
    \multicolumn{1}{l}{\multirow{2}[4]{*}{Technology}} & \multicolumn{3}{c}{Full frontier estimator} &       & \multicolumn{3}{c}{Quantile frontier estimator} \\
\cmidrule{2-4}\cmidrule{6-8}    \multicolumn{1}{l}{} & \multicolumn{1}{c}{MAC} & \multicolumn{1}{c}{$p \times \text{MRT}$} & \multicolumn{1}{c}{$w \times \text{MP}$} &       & \multicolumn{1}{c}{MAC} & \multicolumn{1}{c}{$p\times \text{MRT}$} & \multicolumn{1}{c}{$w \times \text{MP}$} \\
    \midrule
    BP &  $\binom{37.66}{15.16}$    &  $\binom{37.66}{15.16}$ &  $\binom{384475.75}{43989.34}$	 &  &   $\binom{574876.01}{48076.39}$	    &    $\binom{9053462.39}{4535892.31}$   & $\binom{1575827.93}{56992.23}$        \\
    JD &  $\binom{16359.70}{982.40}$	   & $\binom{18559.81}{982.40}$  & $\binom{184227.59}{48711.68}$	 &       & $\binom{5820.69}{24.02}$ & $\binom{7006.28}{24.02}$ & $\binom{90924.64}{57092.26}$     \\
    WGD &  $\binom{8875.65}{9622.23}$  &  $\binom{9005.15}{9622.23}$  &  $\binom{78146.20}{43951.58}$   &    & $\binom{9193.66}{9259.73}$ &  $\binom{9193.66}{9259.73}$	& $\binom{83766.95}{37416.33}$   \\
    \bottomrule
    \end{tabular}%
  \label{tab:tab4}%
\end{table}%

Utilizing the full and quantile frontier estimations, we calculate the MRT between net generation and CO$_2$ emissions and the MP between fuel consumption and CO$_2$ emissions for each power plant. We then derive the monetary shadow prices for CO$_2$ emissions in relation to net generation (i.e., $p \times \text{MRT}$) and total fuel consumption (i.e., $w \times \text{MP}$). The MAC estimate reflects the synergies among abatement options, indicating the minimal electricity loss incurred when reducing fuel consumption to achieve a one-unit reduction in emissions.

Different behavior patterns in MAC estimation are observed across emission-generating technologies when applying full and quantile frontier estimators. The mean and median MACs obtained under the BP technology are larger than those for the other two approaches. For instance, in the CER estimator with BP technology, the MACs have an average value of 574876.01 (\$/1000 tons), while under the JD and WGD technologies, the average MACs are reported as 5820.69 and 9193.66 (\$/1000 tons), respectively.

While reducing fuel consumption can yield long-term savings, it remains a high-cost abatement strategy for the majority of power plants. However, a minority finds it to be the most cost-effective option. The differences between these two abatement alternatives are notably reflected in Table~\ref{tab:tab4}. For instance, quantile estimations indicate that 40\% of power plants consider reducing fuel usage the optimal strategy under the BP technology. In contrast, under JD technology, only 3\% of power plants prefer this alternative. Furthermore, this strategy is deemed cost-ineffective for nearly all power plants under the WGD technology. The findings suggest that downscaling electricity production is a more cost-effective strategy compared to reducing fuel consumption.

In both BP and JD technologies, the average MAC exceeds the corresponding median values under both estimators, indicating a positive-skewed asymmetrical distribution. However, in WGD technology, the opposite is observed (i.e., the mean is less than the median). This aligns with the restrictive summing-up constraint in WGD (i.e., $s_y=0$), leading to the dual price of desirable output (i.e., $\gamma_i$) being set to zero.\footnote{
    To avoid the issue of MRT approaching infinity, we have replaced $\gamma_i=0$ with $\gamma_i=1e-3$.
}
Recall that in a positively skewed distribution, a few outliers shift the mean to the right. 

Fig.~\ref{fig:fig1} presents the characteristics of power plants with the lowest and highest MAC across various emission-generating technologies estimated by the quantile estimator. Fig.~\ref{fig:1a} shows that under the BP technology, power plants with the lowest MAC operate at a large scale, consuming approximately three times the average fuel input and consequently generating higher levels of both emissions and electricity. This elevated level of operation enhances cost-effectiveness, despite higher resource consumption and CO$_2$ emissions. In contrast, power plants with the highest MAC operate at a smaller scale, consuming less fuel (below the sample mean) and producing lower levels of both emissions and electricity. These characteristics are associated with lower efficiency and higher abatement costs. The results suggest that, at higher levels of coal consumption, emission reductions can be achieved more feasibly without substantial losses in electricity generation. Furthermore, the plant with the lowest MAC lies above the 95th quantile, whereas the plant with the highest MAC falls below the 95th quantile in the distribution.
\begin{figure}[htbp]
    \centering
    \begin{subfigure}{0.7\textwidth}
        \centering
        \includegraphics[width=\textwidth]{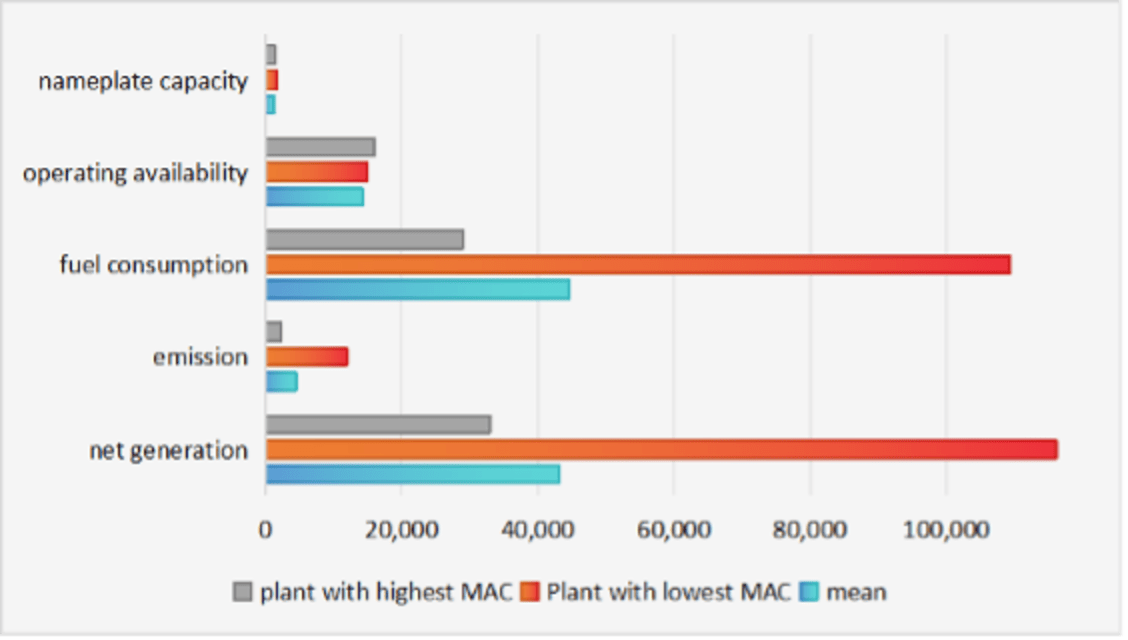}
        \caption{BP}
        \label{fig:1a}
    \end{subfigure}

    \begin{subfigure}{0.7\textwidth}
        \centering
        \includegraphics[width=\textwidth]{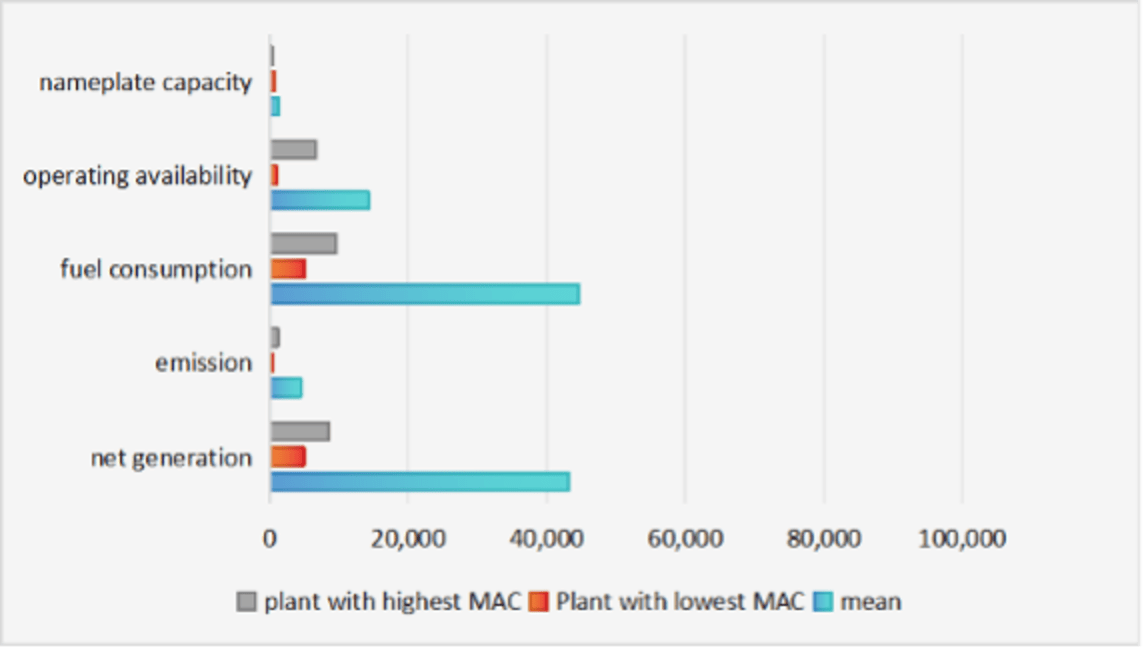}
        \caption{WGD}
        \label{fig:1b}
    \end{subfigure}
    
    \begin{subfigure}{0.7\textwidth}
        \centering
        \includegraphics[width=\textwidth]{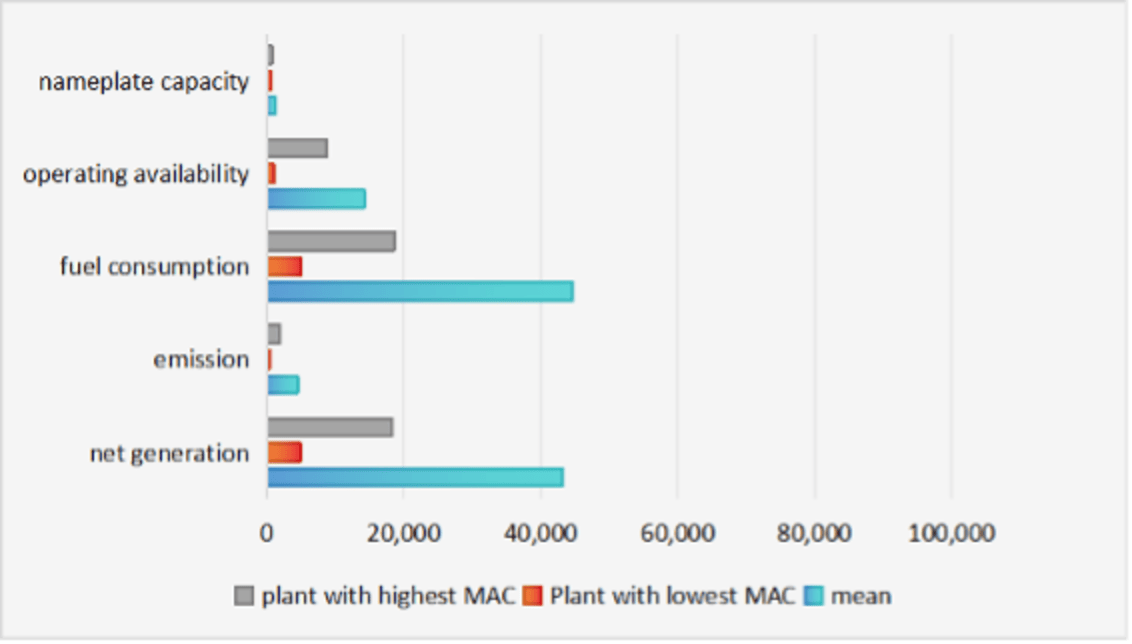}
        \caption{JD}
        \label{fig:1c}
    \end{subfigure}
    \caption{Characteristics of power plants with the lowest and highest MAC across various pollution-generation technologies.}
    \label{fig:fig1}
\end{figure}

Figs.~\ref{fig:1b} and \ref{fig:1c} depict the qualitative characteristics of power plants under JD and WGD technologies, respectively. The patterns observed are broadly consistent across both technologies. Power plants with the lowest MAC operate at a very small scale, using minimal fuel inputs (approximately nine times less than the sample average) and consequently generate lower levels of emissions and electricity. Conversely, plants with the highest MAC also operate on a reduced scale, with fuel consumption roughly two times below the average under JD, and four times below under WGD, leading to correspondingly lower emissions and output. At these lower levels of coal consumption, emission reductions appear more attainable without substantially compromising electricity generation. In addition, power plants with the lowest MAC are located above the 95th quantile, whereas those with the highest MAC fall below the 65th quantile---between the 50th and 65th quantiles for JD, and below the 35th quantile for WGD.

%

\section{Monte Carlo simulation}\label{sec:monte}

To investigate the finite-sample performance of sign-constrained CNLS and CER models under the three emission-generating technologies, we conduct a Monte Carlo study with two distinct scenarios in the absence of noise. The first scenario considers desirable and undesirable outputs separately, while the second scenario integrates them within a unified data-generating process (DGP). By comparing the results in these two scenarios, we can assess the sensitivity of the sign-constrained CNLS and CER methods to the underlying assumptions about production technology.

\subsection{Setup}

The first scenario inherits the conventional BP framework, in which desirable and undesirable outputs are modeled through separate production functions \citep[see, e.g.,][]{Hampf2018, Rodseth2023, Guillen2025}. This specification reflects the traditional view of BP technology, where a Cobb–Douglas production function captures electricity generation, and emissions are assumed to be nearly linearly dependent on fuel consumption. Specifically, consider a general nonparametric production model with observations $\{{x_i^N,x_i^P,b_i,y_i}\}_{i=1}^I$
$$y_i=f(x_i)\times \exp({-u}_{i1})=(x_{i1})^{0.3}({x_{i2})}^{0.3}(x_{i3})^{0.3}\times \exp(-u_{i1}),$$ 
$$b_i=g(x_i^P)\times \exp({-u}_{i2})=0.09404\ (x_{i3})\times \exp(-u_{i2}).$$
where the non-emission-generating inputs (i.e., $x_{i1}$ and $x_{i2}$) and emission-generating input (i.e., $x_{i3}$) are randomly drawn from uniform distribution $x_{im} \sim U[5, 15]$. Both inefficiencies $u_{i1}$ and $u_{i2}$ are independently generated from the positive normal distributions $u_{i1} \sim N^+(0, \sigma^2)$ and $u_{i2} \sim N^+(0,\sigma^2)$, where $\sigma=\{0.3, 0.8, 1.3\}$. 

The second scenario considers that the desirable output ($y$) is a function of both emission-generating and non-emission-generating inputs, as well as the undesirable output, represented by equation $y_i=f(x_i,b_i)$. Specifically, 
$$y_i=f(x_i,b_i)\exp({-u}_i) =(x_{i1})^{0.3}(x_{i2})^{0.3} (x_{i3}{-0.12 b_i)}^{0.3}\exp(-u_i),$$ 
where the inefficiency is randomly drawn from the positive normal distributions $u_i \sim N^+(0, 1.3)$. The negative coefficient of the undesirable output in the DGP confirms that the production of desirable output decreases as the amount of undesirable output increases. This specification is consistent with the simultaneous equation model with the treatment of \citet{Lai2020}.

In Scenario 2, we employ the Cobb–Douglas production function to set up the DGP, due to its well-established theoretical foundation and its suitability for modeling a unified production framework. The Cobb-Douglas specification facilitates the incorporation of emissions via input adjustments, offering a mechanistic interpretation of how emission levels influence desirable output. This formulation enables a direct link between emissions and the effectiveness of emission-generating inputs, aligning closely with our objective to investigate this relationship. Despite the assumption of a constant elasticity of substitution, the Cobb–Douglas function remains an analytically tractable and empirically flexible choice for the present analysis.

To examine the performance of CER models under three emission-generating technologies, the true nonparametric quantile function $Q_y({\tau}|x)$ is defined as  
$$Q_y({\tau}|x)=f(x_i)\times F_{u_{i1}}^{-1}(\tau),$$
where $F_{u_{i1}}^{-1}$ denotes the inverse cumulative normal distribution function for the empirical quantile of the simulated error term $-u_{i1}$. It represents the probability that $-u_{i1}$ takes a value less than or equal to a specified point.

The calculation of the root mean squared error (RMSE) differs fundamentally between the estimation of a production function and that of a quantile production function, due to the different estimation targets involved. In the case of the production function, RMSE captures the average deviation between the estimated function $\left(f_{ir}^{\text{est}} = f_{ir}^{\text{CNLS}}\right)$ and the true production function$f_i$, averaged over $R$ simulation replications. For CER, by contrast, RMSE measures the average deviation between the estimated $\tau$-th conditional quantile function $\left(\hat{Q}_{y_{ir}}(\tau \mid x_i)\right)$ and the true $\tau$-th conditional quantile function$\left(Q_{y_{ir}}(\tau \mid x_i)\right)$, again averaged over $R$ replications. The respective RMSE formulas are as follows:
$$\text{Pro-RMSE}=\frac{1}{R}\sum_{r}{(\frac{1}{I}\sum_{i}{({f_{ir}^{est}-f_i)}^2)}^{0.5}}$$ 
$$\text{Exp-RMSE}=\frac{1}{R}\sum_{r}{(\frac{1}{I}\sum_{i}{({{\hat{Q}}_{y_{ir}}(\tau|x_i)-Q_{y_{ir}}(\tau|x_i))}^2)}^{0.5}}$$

The RMSE serves as a valid performance metric in both contexts; however, its interpretation differs according to the estimation objective. In the case of the sign-constrained CNLS model, RMSE (hereafter referred to as Pro-RMSE) assesses the model's ability to recover the central tendency (i.e., the conditional mean) of the DGP. In contrast, for the CER model, RMSE (denoted as Exp-RMSE) measures the accuracy in estimating a specific quantile of the conditional distribution. It is important to note that RMSE values are always non-negative; thus, smaller values indicate greater estimation accuracy.

\subsection{Simulation results}

In this paper, we set $I = 100$ and $R = 100$, meaning that each scenario consists of 100 randomly generated observations and is replicated 100 times. Tables~\ref{tab:tab5} and \ref{tab:tab6} report the RMSE-based performance of two estimators (sign-constrained CNLS and CER) across the three emission-generating technologies (BP, JD, and WGD), corresponding to Scenario 1 and Scenario 2, respectively.

As shown in Table~\ref{tab:tab5}, the RMSE values associated with the full frontier estimators are substantially higher than those for the quantile frontier estimators. This indicates that recovering the full frontier is inherently more difficult for the sign-constrained CNLS model compared to estimating specific points on the conditional distribution. Furthermore, for the sign-constrained CNLS (across all technologies) and for CER (particularly with WGD and JD technologies), a decrease in RMSE is observed as $\sigma$ increases. This implies that for these estimators and technologies in Scenario 1, a greater dispersion in inefficiency actually leads to a more accurate estimation. Indeed, the multiplicative form of inefficiency, particularly with nonparametric estimators, can lead to complex interactions where a wider spread of inefficiency (higher $\sigma$) provides a better ``signal'' for the estimator to identify the true frontier, thereby reducing bias and overall RMSE.
\begin{table}[H]
  \centering
  \caption{Performance comparison of sign-constrained CNLS and CER models under Scenario 1.}
    \begin{tabular}{lcccc}
    \toprule
          & \multicolumn{1}{l}{\textbf{Technology}} & \textbf{RMSE ($\sigma=0.3$)} & \textbf{RMSE ($\sigma=0.8$)} & \textbf{RMSE ($\sigma=1.3$)} \\
    \midrule
    CNLS          & BP    & 40994.67 & 3823.63 & 33061.09 \\
    CER ($\tau=0.65$)  & \multirow{3}[0]{*}{BP} & 3111.59 & 2866.05 & 2326.92 \\
    CER ($\tau=0.80$)  &       & 2962.75   & 2832.57 &  2456.74 \\
    CER  ($\tau=0.95$) &       &  1981.85 & 1650.24 & 1101.77  \\
    \addlinespace
    CNLS  & JD    & 14118.14 & 7931.92 & 5807.10 \\
    CER ($\tau=0.65$)  & \multirow{3}[0]{*}{JD} &  8787.65 & 8688.77 & 8214.33\\
    CER ($\tau=0.80$)  &       & 6202.42 & 7589.42 & 7745.26 \\
    CER  ($\tau=0.95$) &       & 6146.84 & 6077.13 & 5943.63 \\
    \addlinespace
    CNLS  & WGD   & 5569.97 & 4426.02 & 3438.16 \\
    CER ($\tau=0.65$)  & \multirow{3}[1]{*}{WGD} & 2404.15 & 2026.69 & 1472.51 \\
    CER ($\tau=0.80$)  &       & 2408.96 & 1968.00 &  1624.06 \\
    CER  ($\tau=0.95$) &       & 2350.08 & 2009.73 &  1712.23\\
    \bottomrule
    \end{tabular}%
  \label{tab:tab5}%
\end{table}%

Similarly, the results reported in Table~\ref{tab:tab6} reveal that the CER estimator consistently achieves lower RMSE values for the quantile estimator compared to the sign-constrained CNLS estimator. This highlights the advantage of CER in accurately capturing specific segments of the production distribution. Among the three technologies, WGD clearly outperforms both BP and JD in all levels of inefficiency variation and for both estimators. It consistently yields lower RMSE values, demonstrating superior robustness and accuracy in estimating frontiers or quantiles, even in the presence of undesirable outputs that affect production. The most striking and consistent finding in Scenario 2 is the decrease in RMSE as the standard deviation of inefficiency $(\sigma)$ increases. This applies to CER across all technologies and to CNLS for JD and WGD. This implies that for these methods and technologies, a greater spread or heterogeneity in inefficiency levels actually facilitates more accurate estimation of the true frontier or quantiles in Scenario 2. This could be due to the specific interaction between the model of undesirable output $b_i$ within the production function and the estimation algorithms, where higher variation $u_i$ might provide a clearer signal for distinguishing between the true frontier and the observed outputs.
\begin{table}[H]
  \centering
  \caption{Performance comparison of sign-constrained CNLS and CER models under Scenario 2.}
    \begin{tabular}{lcccc}
    \toprule
          & \multicolumn{1}{l}{\textbf{Technology}} & \textbf{RMSE ($\sigma=0.3$)} & \textbf{RMSE ($\sigma=0.8$)} & \textbf{RMSE ($\sigma=1.3$)} \\
    \midrule
    CNLS & BP    & 66091.46 &  38475.30 & 34134.44 \\
    CER ($\tau=0.65$) & \multirow{3}[0]{*}{BP} & 24321.42 & 15082.81 & 3709.99 \\
    CER ($\tau=0.80$) &       & 17723.70 & 9476.69 & 5022.94 \\
    CER  ($\tau=0.95$) &       & 6407.21 & 5688.78 & 2201.61 \\
    CNLS & JD    & 7201.13 & 3461.68 & 2339.78 \\
    CER ($\tau=0.65$) & \multirow{3}[0]{*}{JD} & 6421.69 & 5445.72 & 4720.64 \\
    CER ($\tau=0.80$) &       & 6090.81 & 4378.22 & 4143.79 \\
    CER  ($\tau=0.95$) &       & 4174.69 & 3379.542 & 2991.18 \\
    CNLS & WGD   & 3418.50 & 313.14 & 22.75 \\
    CER ($\tau=0.65$) & \multirow{3}[1]{*}{WGD} & 44.48 & 21.45 & 15.22 \\
    CER ($\tau=0.80$) &       & 44.31 & 17.73 & 14.51 \\
    CER ($\tau=0.95$) &       & 36.54 & 15.63 & 14.11 \\
    \bottomrule
    \end{tabular}%
  \label{tab:tab6}%
\end{table}%

A significant observation from comparing Table 5 (Scenario 1) and Table 6 (Scenario 2) is the substantially lower RMSE values achieved in Scenario 2, particularly evident with the WGD technology. This improved accuracy is likely attributable to the unified framework in Scenario 2, which explicitly models the joint production of desirable and undesirable outputs and the detrimental effect of the latter. This comprehensive approach appears to enable estimators, especially WGD (which is designed for joint production contexts), to more accurately represent the true underlying production process and its associated inefficiencies. Furthermore, the counterintuitive trend of RMSE decreasing with increasing $\sigma$ is present in both scenarios, but it is more pronounced and consistent in Scenario 2. This suggests that the unified framework, where the undesirable output directly influences the production function, provides a clearer ``signal'' for the estimators when inefficiency is more dispersed. This enhanced signal may allow the estimators to better disentangle the true frontier from the combined effects of inputs, undesirable outputs, and varying inefficiency. From a technological perspective, the consistent superiority of WGD across both scenarios is particularly noteworthy given the distinct ways the undesirable output is handled. This confirms the robustness and suitability of WGD for modeling production processes that generate undesirable outputs, whether they are formulated separately or within a unified framework.

%

\section{Conclusions}\label{sec:concl}

Coal-fired power plants play a pivotal role in U.S. electricity generation, making this industry central to achieving national carbon reduction goals and facilitating the broader transition away from fossil fuels in other sectors. Establishing effective and economically viable emission reduction targets and policies requires a clear understanding of the costs associated with mitigating CO$_2$ emissions. This study addresses this need by estimating the MAC of CO$_2$ emission reduction for U.S. coal-fired power plants, using plant-level CO$_2$ emission data in conjunction with corresponding financial records.

We start by reviewing several BP, JD, and WGD approaches for modeling emission-generating technologies. To estimate the MAC of CO$_2$ emissions, we apply both full and quantile frontier estimators. The key methodological distinction lies in the treatment of inefficiency: the full frontier estimator inherently disregards inefficiency in shadow price calculation, whereas the quantile frontier estimator employs nearest quantiles, explicitly taking inefficiency into account and producing more precise and robust results. Instead of following the conventional MAC measure, which assesses the loss of desirable output per unit of emission reduction under fixed inputs and may not identify the least-cost strategy, we illustrate potential firm-level abatement strategies. Consistent with EPA data showing a typical fuel mix of approximately 92.59\% coal, 6.34\% natural gas, and 0.73\% oil for bituminous coal-fired power plants, our analysis recognizes the limited role of auxiliary fuels. Given that the EPA reports only total fuel consumption, we focus on reducing production and lowering fuel consumption as the dominant abatement strategies, while excluding fuel switching from consideration.

The findings show that MACs under BP technology are substantially higher than those under the other two technologies. The optimality of the cost-effective abatement strategy varies markedly across technologies. Under BP technology, about 40\% of plants favor fuel reduction, whereas only a negligible proportion do so under JD technology and almost none under WGD technology. As a result, downscaling production emerges as a more economically viable strategy. The analysis of plant characteristics indicates that, for BP technology, economies of scale are an important driver of cost-effectiveness: plants operating at a high scale record the lowest MAC, while those at a small scale have the highest. In contrast, for both JD and WGD technologies, the lowest MACs are associated with very small-scale operations, suggesting that emission reductions are more readily achieved at reduced production levels. This difference underscores the technology-specific nature of cost-effective emission abatement strategies. Within a tradable permit system, plants with low MAC are encouraged to sell permits, thereby contributing to cost-effective overall emission reductions, while plants with high MAC are encouraged to purchase permits at a mutually agreed market price for emissions, allowing flexibility in meeting regulatory requirements and promoting broader abatement efforts.

We further employ Monte Carlo simulations to evaluate and compare the performance of emission-generating technologies under two distinct DGPs. The results indicate that the quantile frontier estimation (CER) generally produces more accurate parameter estimates than the conventional full frontier method (sign-constrained CNLS) in both scenarios. The WGD technology consistently achieves the highest accuracy across both scenarios and estimation tasks. Although the quantile estimation generally performs better than the full frontier method, the simulations reveal a recurring and counterintuitive pattern: greater variation in inefficiency often leads to improved estimation accuracy, as measured by lower RMSE, for the more robust technologies and estimators. This effect is especially pronounced in Scenario 2, where undesirable outputs are explicitly modeled within a unified framework.

For future research, the evaluation of MAC will not be limited solely to CO$_2$ emissions. Determining the MAC associated with other commonly observed pollutants, such as SO$_2$ and NOx, would warrant further investigation, potentially involving the development of two-stage network modeling structures to analyze end-of-pipe abatement technologies, such as flue gas desulfurization for SO$_2$ emissions. Future work can also develop a comprehensive framework that incorporates input-switching and the integration of renewable energy sources as potential strategies for reducing emissions in electricity generation. This will entail examining the adoption of alternative, lower-emission fuels and incorporating renewable energy sources, such as solar or wind power, into production practices. Such a framework can adapt DDF to explicitly account for changes in input mixes and energy sources, thereby enabling a rigorous assessment of their impact on both productivity and environmental performance. Finally, the empirical findings of this paper can provide a broader perspective on emission-generating technologies within the year 2022. Future research applying this analytical approach over a longer temporal scale would be particularly beneficial for the formulation of comprehensive and impactful environmental regulations.

%

\baselineskip 12pt
\bibliographystyle{econ-econometrica}
\bibliography{References.bib} 

%

\newpage
\baselineskip 20pt
\section*{Appendix}\label{sec:app}

\renewcommand{\thesubsection}{A\arabic{subsection}}
\renewcommand{\theequation} {A\arabic{equation}}
\setcounter{equation}{0}

\subsection{Proof of Proposition 1} \label{App:proof1}
To derive the shadow prices under the BP technology, we formulate the dual of the linear programming problem~\eqref{eq:eq5} as follows:
\vspace{-0.35cm}
\begin{alignat}{2}
	\mathop{\min }\, \quad & (-\gamma_i^\prime y_{oj}+\beta_i^\prime x_{om}^N+\eta_i^\prime x_{om}^p+\omega_i^\prime b_{ok}+\alpha_i-{\bar{\alpha}}_i) &{\quad}& \label{eq:a1}\\
	\mbox{s.t.}\quad
    & -\gamma_i^\prime y_{ij}+\beta_i^\prime x_{im}^N+(\eta_i^\prime+{\bar{\eta}}_i^\prime{)x}_{im}^p+\alpha_i\geq0 &{}& \forall i\notag \\
	& \omega_i^\prime b_{ik}-({\bar{\eta}}_i^\prime x_{im}^p+{\bar{\alpha}}_i)\geq0  &{}& \forall i, k, m \notag \\
    & \gamma_i^\prime g_y \geq 0.5  &{}& \forall i \notag \\
	& \omega_i^\prime g_b \geq 0.5 &{}& \forall i \notag \\ 
    & \eta_i^\prime g_x \geq 0  &{}& \forall i \notag \\
    & \gamma_i, \beta_i,\eta_i,  \omega_i \geq 0  &{}& \forall i \notag  \\
    & \alpha_i, \bar{\alpha}_i, \bar{\eta}_i \, \text{free}  &{}& \forall i \notag  
\end{alignat}

We then introduce $\varepsilon_i^1$ , $\varepsilon_i^2$ and $\varepsilon_i$ as auxiliary variables as follows:

$\varepsilon_i^1=[\alpha_i+\beta_i^\prime x_{im}^N+(\eta_i^\prime+{\bar{\eta}}_i^\prime{)x}_{im}^p]-\gamma_i^\prime y_{ij}; \varepsilon_i^1 \geq0,$

$\varepsilon_i^2=\omega_i^\prime b_{ik}-[{\bar{\alpha}}_i+{\bar{\eta}}_i^\prime x_{im}^p]; \varepsilon_i^2\geq0,$

$\varepsilon_i=\varepsilon_i^1+\varepsilon_i^2=(\alpha_i-{\bar{\alpha}}_i)+\beta_i^\prime x_{im}^N+\eta_i^\prime x_{im}^p+\omega_i^\prime b_{ik}-\gamma_i^\prime y_{ij}\geq0;\ \varepsilon_i\geq0.$

Since the inefficient firm does not affect the shape of the sub-technologies, we introduce $\varepsilon_i^1 \geq 0$ and $\varepsilon_i^2 \geq 0$ on the left- or right-hand side of the constraints, respectively. This adjustment induces the concavity of the supporting hyperplanes for the two sub-technologies:

$\alpha_h{+\beta}_h^\prime x_{im}^N+(\eta_h^\prime+{\bar{\eta}}_h^\prime{)x}_{im}^p-\gamma_h^\prime y_{ij}\le\alpha_i{+\beta}_i^\prime x_{im}^N+(\eta_i^\prime+{\bar{\eta}}_i^\prime{)x}_{im}^p-\gamma_i^\prime y_{ij}$

${\omega_h^\prime b_{ik}-(\bar{\alpha}}_h+{\bar{\eta}}_h^\prime x_{im}^p)\le\ \omega_i^\prime b_{ik}-({\bar{\alpha}}_i+{\bar{\eta}}_i^\prime x_{im}^p)$

This leads to the sign-constrained CNLS model~\eqref{eq:eq6} under BP technology, which accounts for the trade-off between desirable and undesirable outputs. 
\end{document}